\documentclass[11pt]{amsart} 
\usepackage{amscd,amssymb,amsxtra}
\usepackage[mathscr]{eucal}
\usepackage{mathabx}
\usepackage{comment}
\usepackage{color}
\usepackage{enumitem} 
\makeatletter
\let\@secnumfont\bfseries
\def\section{\@startsection{section}{1}%
  \z@{4\linespacing\@plus\linespacing}{\linespacing}%
  {\bfseries\centering}}
\def\introsection{\@startsection{section}{1}%
  \z@{3\linespacing\@plus\linespacing}{\linespacing}%
  {\bfseries\centering}}
\def\subsection{\@startsection{subsection}{2}%
   \z@{1.25\linespacing\@plus.7\linespacing}{.5\linespacing}%
   {\normalfont\bfseries}}
\def\subsectionsinline{\def\subsection{\@startsection{subsection}{2}%
  \z@{1\linespacing\@plus.7\linespacing}{-.5em}%
  {\normalfont\bfseries}}}

\makeatother

\theoremstyle{definition}

\newtheorem*{definition*}{Definition}
\newtheorem*{example*}{Example}
\newtheorem*{problem*}{Problem}
\newtheorem*{exercise*}{Exercise}
\newtheorem*{question*}{\color{blue}Question}
\newtheorem*{construction*}{Construction}

\theoremstyle{remark}

\newtheorem{example}[equation]{Example}
\newtheorem{remark}[equation]{Remark}

\newtheorem*{note*}{Note}
\newtheorem*{notation*}{Notation}
\newtheorem*{remark*}{Remark}
\newtheorem*{data*}{Data}

\theoremstyle{plain}

\newtheorem*{theorem*}{Theorem}
\newtheorem*{corollary*}{Corollary}
\newtheorem*{lemma*}{Lemma}
\newtheorem*{proposition*}{Proposition}
\newtheorem*{conjecture*}{Conjecture}
\newtheorem*{claim*}{Claim}
\newtheorem*{proposal*}{Proposal}
\newtheorem*{conclusion*}{Conclusion}
\newtheorem*{hypothesis*}{Hypothesis}
\newtheorem*{assumption*}{Assumption}

\newenvironment{proof*}[1][\proofname]{
  \begin{proof}[#1]}{  
\end{proof}}

\numberwithin{equation}{subsection}

\definecolor{refkey}{rgb}{0,.6,.4}

\newcommand{\CC}{{\mathbb C}}

\newcommand{\EE}{\mathbb E}

\DeclareMathOperator{\pt}{pt}

\newcommand{\RP}{{\mathbb R\mathbb P}}
\newcommand{\RR}{{\mathbb R}}

\DeclareMathOperator{\Spin}{Spin}

\newcommand{\ZZ}{{\mathbb Z}}

\newcommand{\chiup}{\raise.5ex\hbox{$\chi$}}
\newcommand{\cir}{S^1}

\newcommand{\inv}{^{-1}}
\DeclareRobustCommand{\mstrut}{^{\vphantom{1*\prime y\vee M}}}

\newcommand{\temsquare}{\raise3.5pt\hbox{\boxed{ }}}

\newcommand{\zmod}[1]{\ZZ/#1\ZZ}

\newcommand{\zt}{\zmod2}



\numberwithin{subsection}{section}

\usepackage[all,2cell]{xy}\renewcommand{\cir}{\ensuremath{S^1}}
\usepackage{graphicx}
\usepackage{epsf}

\DeclareMathOperator{\Arf}{Arf}
\DeclareMathOperator{\Cliff}{Cliff}
\DeclareMathOperator{\Thom}{Thom}
\newcommand{\EG}{E}
\newcommand{\IZ}{I\ZZ}
\newcommand{\MSpin}{M\!\Spin}
\newcommand{\bU}{\overline{U}}
\newcommand{\bY}{\overline{Y}}
\newcommand{\bdY}{\partial Y}
\newcommand{\gpd}{/\!/}  
\newcommand{\tE}{\widetilde{E}}
\DeclareMathOperator{\spin}{Spin}

\DeclareMathOperator{\mspin}{MSpin}
\DeclareMathOperator{\mtspin}{MTSpin}
\DeclareMathOperator{\thom}{Thom}

\DeclareMathOperator{\Lie}{Lie}

\newcommand{\E}{\mathbb{E}}
\newcommand{\rp}{\mathbf{RP}}
\newcommand{\rnstar}[1]{{\R^{\1}}^{\ast}}

\newcommand{\slot}{\,-\,}

\newcounter{textItem}

\newcommand{\textItemref}[1]
 	 {{\rm (\ref{#1})}}

\newenvironment{textList}{\begin{list}%
{\rm (\Alph{textItem})}{\usecounter{textItem}
\setlength{\labelwidth}{2em}
\setlength{\itemindent}{2em}
\setlength{\leftmargin}{0pt}
\setlength{\listparindent}{0pt}
\setlength{\parsep}{0pt}
\setlength{\partopsep}{0pt}
\setlength{\itemsep}{\medskipamount}
\setlength{\topsep}{\medskipamount}
}}{\end{list}}

	\theoremstyle{remark}

\newcommand{\Z}{{\mathbb Z}}
\newcommand{\R}{{\mathbb R}}

\newtheorem{prop}[equation]{Proposition}
\newtheorem{clm}[equation]{Claim}
\newtheorem*{thm*}{Theorem}
\newtheorem*{cor*}{Corollary}
\newtheorem*{lem*}{Lemma}
\newtheorem*{prop*}{Proposition}

\begin{document}

\abovedisplayskip18pt plus4.5pt minus9pt
\belowdisplayskip \abovedisplayskip
\abovedisplayshortskip0pt plus4.5pt
\belowdisplayshortskip10.5pt plus4.5pt minus6pt
\baselineskip=15 truept
\marginparwidth=55pt

\makeatletter
\renewcommand{\tocsection}[3]{%
  \indentlabel{\@ifempty{#2}{\hskip1.5em}{\ignorespaces#1 #2.\;\;}}#3}
\renewcommand{\tocsubsection}[3]{%
  \indentlabel{\@ifempty{#2}{\hskip 2.5em}{\hskip 2.5em\ignorespaces#1%
    #2.\;\;}}#3} 
\renewcommand{\tocsubsubsection}[3]{%
  \indentlabel{\@ifempty{#2}{\hskip 2.5em}{\hskip 4.7em\ignorespaces#1%
    #2.\;\;}}#3} 
\makeatother

\setcounter{tocdepth}{3}

\renewcommand{\thesubsection}{\bf\arabic{subsection}}
\renewcommand{\theequation}{\arabic{subsection}.\arabic{equation}}

\theoremstyle{definition}
\newtheorem{ansatz}[equation]{Ansatz}



 \title[Invertible Phases with Spatial Symmetry]{Invertible Phases of Matter with Spatial Symmetry} 
 \author[D. S. Freed]{Daniel S.~Freed}
 \address{Department of Mathematics \\ University of Texas \\ Austin, TX
78712} 
 \email{dafr@math.utexas.edu}
 \author[M. J. Hopkins]{Michael J.~Hopkins}
 \address{Department of Mathematics \\ Harvard University \\ Cambridge, MA
02138} 
 \email{mjh@math.harvard.edu}
 \thanks{This material is based upon work supported by the National Science
Foundation under Grant Numbers DMS-1158983, DMS-1160461, DMS-1510417, and
DMS-1611957.  Any opinions, findings, and conclusions or recommendations
expressed in this material are those of the authors and do not necessarily
reflect the views of the National Science Foundation.}
 \date{\today}
 \begin{abstract} 
 We propose a general formula for the group of invertible topological phases
on a space~$Y$, possibly equipped with the action of a group~$G$.  Our
formula applies to arbitrary symmetry types.  When $Y$~is Euclidean space and
$G$~a crystallographic group, the term `topological crystalline phases' is
sometimes used for these phases of matter.
 \end{abstract}
\maketitle


In previous work~\cite{FH}, recalled in~\S\ref{subsec:1.1} below, we
determine the homotopy type of the space of invertible field theories with a
fixed symmetry type.  This result is a theorem about \emph{field theories} in
the framework of the Axiom System for field theory introduced by Segal in the
1980's.  It has wide applicability: invertible field theories enter quantum
field theory and string theory in many different ways.  In condensed matter
theory our theorem can be used to classify invertible phases of matter (on
Euclidean space), but only accepting standard unproved assertions about
effective low energy field theories of discrete models.  In this note we
combine this theorem with a few more basic principles~(\S\ref{subsec:1.2}) to
offer a general formula for the abelian group of invertible topological
phases of matter on a topological space~$Y$ equipped with the action of a
group~$G$. \emph{Time} does not appear: $Y$~models \emph{space}, not
spacetime.  We motivate and present the formula in Ansatz~\ref{thm:1} and
Ansatz~\ref{thm:3}; the formula depends on a symmetry type but not on a
dimension.  As evidence we compute some illustrative examples and compare to
known results.  (See Example~\ref{thm:7} and Example~\ref{thm:10}.)  The
pedagogical aspirations of this note are realized in~\S\ref{subsec:1.7},
where we briefly explain some computational techniques in Borel equivariant
homotopy theory, and in~\S\ref{subsec:1.8}, where we illustrate via a
specific example---a half-turn in 3-space---which we attack using three
different methods.
 
The idea that invertible phases comprise a generalized \emph{homology} group
on space was suggested by Alexei Kitaev; he works with lattice models to
motivate the particular homology theory.  There are discussions of special
cases of the problem we treat here in~\cite{SHFH,TE,HSHH}.  The recent
paper~\cite{SXG} uses a spectral sequence to compute the group of phases, as
do we in~\S\ref{subsubsec:1.7.2}, \S\ref{subsubsec:1.8.3}, but the
generalized homology theory is not specified and physical arguments are used
to compute differentials.  We thank Lukasz Fidkowski, Mike Hermele, and
Ashvin Vishwanath for bringing the specific example treated
in~\S\ref{subsec:1.8} and the general problem to our attention, as well as for
a very informative email correspondence.

{\small
\def\reftext{References}
\renewcommand{\tocsection}[3]{%
  \begingroup 
   \def\tmp{#3}%
   \ifx\tmp\reftext
  \indentlabel{\qquad\quad\;\, } #3%
  \else\indentlabel{\ignorespaces#1 #2.\;\;}#3%
  \fi\endgroup}
\tableofcontents
}

  \subsection{Recollection of~\cite{FH}}\label{subsec:1.1}

Let $d$~be the dimension of space.  The \emph{symmetry type} of a
Wick-rotated relativistic field theory in spacetime dimension~$d+1$ is
described by a pair~$(H,\rho )$.  The topological group~$H$ is the colimit of
a sequence of compact Lie groups~$H_{d+1}$, each sitting in a group extension
  \begin{equation}\label{eq:1}
     1\longrightarrow K\longrightarrow H_{d+1}\xrightarrow{\;\;\rho
     _{d+1}\;\;} O_{d+1} 
  \end{equation}
in which the image of~$\rho _{d+1}$ is either~$O_{d+1}$ (symmetry type with
time-reversal) or~$SO_{d+1}$ (no time-reversal).  Then $\rho \:H\to O$ is the
stabilization of~$\rho _{d+1}$ as~$d\to\infty $; see~\cite[\S2]{FH}.  The
subgroup~$K$ is the group of internal symmetries---those which act trivially
on \emph{spacetime}---and is independent of~$d$.  (If we break relativistic
invariance, there is a slightly larger group which acts trivially on
\emph{space}; see~\cite[Remark~9.32]{FH}.)  The homomorphism~$\rho $
determines a rank zero virtual real vector bundle $W\to BH$, the
stabilization of rank zero virtual bundles over~$BH_{d+1}$, and there is a
corresponding Thom spectrum
  \begin{equation}\label{eq:2}
     MTH = \Thom(BH;-W) 
  \end{equation}
of the virtual vector bundle $-W\to BH$.  Let $I\ZZ$~be the Anderson dual
to the sphere spectrum and
  \begin{equation}\label{eq:3}
     E = E_{(H,\rho )} = \Sigma ^2\IZ^{MTH} 
  \end{equation}
the spectrum of maps $MTH\to\Sigma ^2\IZ$.  Then the main
outcome\footnote{This statement is left as a conjecture in that paper; what
is proved from various ans\"atze is an identification of the \emph{torsion}
subgroup with isomorphism classes of invertible \emph{topological} theories.
The entire group~\eqref{eq:5} is also proved to be the group of isomorphism
classes of ``continuous'' invertible theories; see \cite[\S5.4]{FH}.}
of~\cite{FH} is an identification of
  \begin{equation}\label{eq:5}
     E_{-d}(\pt)\cong E^d(\pt) \cong [MTH,\Sigma ^{d+2}\IZ]
  \end{equation}
as the group of deformation classes of invertible reflection positive
extended field theories in~$d+1$ dimensions with symmetry type~$(H,\rho )$.
Computations for various~$(H,\rho )$ may be found in \cite[\S\S9--10]{FH} as
well as \cite{Ka,KTTW,C,BC,GPW}.

  \subsection{Invertible phases on a space}\label{subsec:1.2}
 
We imagine that invertible topological phases can be localized in space,
possibly with noncompact support; satisfy some locality properties; and are
equipped with a pushforward under proper continuous maps.  Since
$E_0(\pt)$~is the group of invertible phases in $0+1$ dimensions---that is,
phases on a point---we posit the following.

  \begin{ansatz}[]\label{thm:1}
 Let $Y$~be a locally compact topological space.  Then the group of
invertible topological phases on~$Y$ of symmetry type~$(H,\rho )$ is the
Borel-Moore homology group~$E\mstrut _{0,BM}(Y)$.
  \end{ansatz}

\noindent
 If $Y$~is the complement in a finite CW complex~$\bY$ of a
subcomplex~$Y_0\subset \bY$, then Borel-Moore homology reduces to relative
homology: $E\mstrut _{0,BM}(Y)\cong E_0(\bY,Y_0)$.  Thus on Euclidean $d$-space we
have
  \begin{equation}\label{eq:4}
     E\mstrut _{0,BM}(\EE^d)\cong E_0(S^d,\pt)\cong E_{-d}(\pt),
  \end{equation}
which recovers~\eqref{eq:5}.  If $Y$~is compact, then $E\mstrut
_{0,BM}(Y)\cong E_0(Y)$.

  \begin{example}[Phases on a torus]\label{thm:7}
 Let $Y=(\cir)^{\times d}$ be the $d$-dimensional torus.  After suspension
$Y$~is homotopy equivalent to a wedge of spheres, from which 
  \begin{equation}\label{eq:6}
     E_0(Y)\cong \bigoplus\limits_{i=0}^d \;E_{-i}(\pt)^{\oplus {d\choose i}}\;. 
  \end{equation}
For example, if $d=2$ and we consider fermionic theories ($H=\Spin$), then 
  \begin{equation}\label{eq:7}
     E_0(S^1\times \cir)\cong (\zt)\;\oplus\; (\zt\oplus \zt)\;\oplus\; (\ZZ);
  \end{equation}
the summands correspond to theories supported on a point, on the 1-cells
(figure eight), and on the 2-cell, respectively.  We remark that all classes
are represented by free fermions: \eqref{eq:7}~is also isomorphic
to~$KO^0(S^1\times \cir)$.  See~\cite{R} for a discussion of the physics of
this example.
  \end{example}

  \begin{remark}[Invertible phases on a compact smooth manifold]\label{thm:8}
 A compact smooth $d$-manifold~$Y$ with boundary has a Spanier-Whitehead
dual~$D(Y/\bdY)\simeq \Thom(Y;-TY)\simeq \Sigma ^{-TY}Y$, according
to~\cite{A}, and so
  \begin{equation}\label{eq:8}
    \begin{aligned}
     E_0(Y,\bdY)&\cong [S^0\,,\,E\wedge Y/\bdY] \\&\cong [MTH\,,\,\Sigma
     ^2\IZ\wedge Y/\bdY]\\&\cong 
     [\Sigma ^{\underline{\RR^d}-TY}(Y)\wedge MTH\,,\,\Sigma ^{d+2}\IZ]\\&\cong
     [\Sigma ^{\underline{\RR^d}-TY}(Y)\,,\,\Sigma ^dE]\,,\, \end{aligned}
  \end{equation}
where $\underline{\RR^d}\to Y$ is the trivial vector bundle with
fiber~$\RR^d$.  This last group is a twisted $E$-cohomology group of~$Y$; the
twisting is trivialized by an $E$-orientation of~$Y$.   
 
The third line of~\eqref{eq:8} may be regarded as deformation classes of
invertible field theories of symmetry type~$(H,\rho )$ with a background
scalar field valued in~$Y$, or rather in a twist of~$Y$ if $Y$~is not
$E$-oriented.  This field theory interpretation was used in~\cite{TE} to
study special cases.
  \end{remark}

  \subsection{Invertible phases on a $G$-space}\label{subsec:1.4}

It is natural consider a compact Lie group $G$ is acting on a locally compact
space $Y$ and model {\em equivariant phases} on $Y$.\footnote{We allow
noncompact groups acting with compact isotropy subgroups, i.e., topological
stacks with compact Lie group stabilizers~\cite[A.2.2]{FHT}.
Example~\ref{thm:7} is of this type: $(\cir)^{\times d}$~is isomorphic to the
quotient stack $\EE^d\gpd \ZZ^d$.}  For this there is a choice to make and so
far simply working with {\em Borel equivariant} homotopy theory seems to
work.  We therefore work in the category of Borel $G$-equivariant spectra.
See~\cite[\S6]{FH} for an introduction and for notation explanation.  We
write $[\slot,\slot]^{hG}$ for the abelian group of homotopy classes of Borel
equivariant maps between $G$-spectra.  

As evidence in favor of Borel equivariant spectra, consider the case when
$Y^d$~is a closed manifold and $G$~acts trivially on~$Y$.  Interpret the last
line of~\eqref{eq:8} as twisted $E$-cohomology; replace $E$-cohomology by
Borel equivariant $E$-cohomology; use the fact that the Borel $G$-equivariant
cohomology of~$Y$ is the nonequivariant $E$-cohomology of the Borel
construction $EG\times _GY$; then since $G$~acts trivially on~$Y$, the Borel
construction reduces to $EG\times _GY\cong BG\times Y$; hence the Borel
equivariant version of~\eqref{eq:8} is
  \begin{equation}\label{eq:9}
     [\Sigma ^{\underline{\RR^d}-TY}(Y)\wedge MTH\wedge BG_+\,,\,\Sigma ^{d+2}\IZ]
     \cong \tE_0(Y), 
  \end{equation}
where $\tE$~is the spectrum~\eqref{eq:3} for the symmetry type $(H\times
G,\rho \times e)$ obtained from~$(H,\rho )$ by taking the Cartesian product
with~ $G$ as an internal symmetry.  This is the expected answer.  

Denote the Borel equivariant homology of a $G$-space~$Y$ as
  \begin{equation}\label{eq:18}
     E^{hG} _0(Y) = [S^0,E\wedge Y_+]^{hG},
  \end{equation}
where on the right hand side $E$~is regarded as a $G$-spectrum with trivial
$G$-action. 

  \begin{ansatz}[]\label{thm:3}
 Let $Y$~be a locally compact topological space equipped with the action of a
compact Lie group~$G$.  Then the group of invertible topological phases
on~$Y$ of symmetry type~$(H,\rho )$ is the Borel-Moore equivariant
\emph{homology} group $E^{hG}_{0,BM}(Y)$.
  \end{ansatz}

  \begin{remark}[]\label{thm:9}
 Whereas Borel equivariant $E$-cohomology is the $E$-cohomology of the Borel
construction, Borel equivariant $E$-homology~\eqref{eq:18} is \emph{not} the
$E$-homology of the Borel construction.
  \end{remark}

  \begin{example}[Euclidean symmetries with a fixed point]\label{thm:10}
 Suppose $Y=\EE^d$ and $G$~is a group of isometries which fixes a point~$p\in
\EE^d$.  Use $p$~as a basepoint to identify the affine space~$\EE^d$ with the
vector space~ $\RR^d$; then the action is described by a homomorphism
$\lambda \:G\to O_d$.  Let $S^\lambda $ denote the associated representation
sphere: the one point compactification of~$\RR^d$ with basepoint the new
point at infinity and inherited $G$-action.  Then Ansatz~\ref{thm:3} computes
the group of invertible phases:
  \begin{equation}\label{eq:10}
     \begin{aligned} E^{hG}_0(S^d\,,\,\infty )&\cong [S^0\,,\,\EG\wedge
     S^\lambda ]^G 
      \\ &\cong [S^{-\lambda }\,,\,\EG]^G \\ &\cong [\Sigma ^{d-\lambda }(BG)
      \wedge MTH\,,\,\Sigma ^{d+2}\IZ] \\ &\cong
      \bigl[\Thom(BH\times
     BG;-W+\underline{\RR^d}-V_\lambda )\,,\,\Sigma ^{d+2}\IZ
     \bigr],\end{aligned}  
  \end{equation}
where $V_\lambda \to BG$ is associated to~$\lambda $.  (The
isomorphism~\eqref{eq:10} is a special case of~\eqref{eq:8}.)  The last
expression in~\eqref{eq:10} is the group of invertible phases in $d$~space
dimensions of the symmetry type $(H\times G,\rho \times \lambda )$.  For
$H=SO$ (bosonic theories) this reduces to the ``crystalline equivalence
principle'' of~\cite{TE} in dimensions~$d\le 1$ for which we can
replace~$MSO$ by~$H\ZZ$.  (Note that \eqref{eq:10} includes a twist for
symmetries which reverse orientation.)
  \end{example}

  \subsection{Computational techniques in Borel equivariant theory}\label{subsec:1.7}

We offer a brief exposition of computational methods, relying on
\cite[\S6]{FH} and the references therein for background on equivariant
stable homotopy theory.

   \subsubsection{Reduction to nonequivariant computations}\label{subsubsec:1.7.1}
 The evaluation of the Borel equivariant maps between $G$-spectra can often be
reduced to the computation of non-equivariant maps by the following devices.

\begin{textList}
\item When $M$ is a $G$-spectrum and $N$ is an ordinary spectrum,
regarded as a $G$-spectrum with trivial action one has  
  \begin{equation}\label{eq:21}
     [M,N]^{hG}=[EG_{+}\underset{G}{\wedge} M,N]. 
  \end{equation}

\item\label{item:1} ({\em Adams isomorphism}).  When $M$ has trivial
$G$-action, $N$~is a $G$-spectrum, and $T$ is a finite free $G$-CW-complex,
the transfer map  
  \begin{equation}\label{eq:20}
     [M, (N\wedge T_{+}\wedge S^{\mathfrak{g}})_{hG}] \to [M,N\wedge
     T_{+}]^{hG} 
  \end{equation}
is an isomorphism.  Here $S^{\mathfrak{g}}$ is the
one point compactification of the Lie algebra of $G$ and  
  \[
     (N\wedge S^{\mathfrak{g}})_{hG} = EG_{+}\underset{G}{\wedge}(N\wedge
     S^{\mathfrak{g}}). 
  \]

\item Atiyah duality identifies the Spanier-Whitehead dual of a closed
manifold $M$ with the Thom complex $M^{-TM}$.  When $W\subset G$ is a
closed subgroup this implies that the Spanier-Whitehead dual of the
homogeneous space $G/W$ is the Thom spectrum
$G_{+}\underset{W}{\wedge}S^{-\mathfrak{g}/\mathfrak{w}}$, in which
$\mathfrak{g}=\Lie G$ and $\mathfrak{w}=\Lie W$.

\item When $W\subset G$ is a closed subgroup one has an isomorphism
\[
[M\wedge G_{+}/W,N]^{hG} = [M,N]^{hW} 
\]
from which, using Atiyah duality, one deduces an isomorphism
\[
[M,N\wedge G/W_{+}]^{hG} = [M,N\wedge S^{\mathfrak{g}/\mathfrak{w}}]^{hW}.
\]
\end{textList}

\begin{remark}
\label{rem:2} In~\textItemref{item:1} when $N$ is the suspension
spectrum of a $G$-space $X$ then $(N\wedge S^{\mathfrak{g}})_{hG}$ is the
suspension spectrum of the Thom complex
\[
\thom(EG\underset{G}{\times}X,\mathfrak{g}).  
\]
\end{remark}


Computations in Borel equivariant homotopy theory can be made using
the above rules, augmented with knowledge of the effect of the maps
\begin{align}
\label{eq:m11}
[M,N\wedge G/(W_{1})_{+}]^{hG} &\to [M,N\wedge G/(W_{2})_{+}]^{hG} \\ 
\label{eq:m12}
[M\wedge G/(W_{2})_{+},N]^{hG} &\to [M\wedge G/(W_{1})_{+},N]^{hG} 
\end{align}
induced by an equivariant map 
\[
G/W_{1}\to G/W_{2}.
\]

\begin{remark}
\label{rem:5}
In the extended example in~\S\ref{subsec:1.8}, the group $G$ is cyclic of
order $2$ and the only map whose effect need be worked out is
\[
G\to G/G.
\]
When $M$ and $N$ have trivial $G$-action, the maps~\eqref{eq:m11}
and~\eqref{eq:m12} are identified, using the rules above, with the maps 
\begin{align*}
[M, N] &\to [M\wedge BG_{+},N] \\ 
[M\wedge BG_{+}, N] &\to [M,N] \\ 
\end{align*}
induced by the transfer map $BG_{+}\to S^{0}$ and the
map $S^{0}\to BG_{+}$ associated to a choice of point in $BG$.
\end{remark}

   \subsubsection{Equivariant Atiyah-Hirzebruch spectral sequence}\label{subsubsec:1.7.2}
To motivate the construction assume $G$~is a finite group and $Z$~a pointed
$G$-space.  Let $L'\subset G$ be a subgroup and suppose $f\:G/L'\times
S^{p-1} \to Z$ is a continuous $G$-equivariant map for some positive
integer~$p$.  The mapping cone of~$f$ is the union $W=Z\cup_f(G/L'\times
D^p)$ which attaches an equivariant $p$-cell to the space~$Z$.  From the
equivariant cofibration sequence
  \begin{equation}\label{eq:15}
     Z\longrightarrow W\longrightarrow W/Z\simeq G/L'\times (D^p,S^{p-1})
  \end{equation}
we obtain a boundary map in equivariant homology: 
  \begin{equation}\label{eq:16}
     \partial \:E^{hG}_k(W,Z)\longrightarrow E^{hG}_{k-1}(Z) .
  \end{equation}
By excision and~\eqref{eq:21} the domain is isomorphic to $E^{p-k}(BL')$,
which by~\eqref{eq:5} is interpreted as a group of topological phases in
spatial dimension~$p-k$.  If $E=E_{(H,\rho )}$ as in~\eqref{eq:3}, then these
theories have symmetry type $(H\times L',\rho \times e)$.  Suppose $Z$~is
obtained from a subcomplex~$Z'\subset Z$ by attaching an equivariant
$(p-1)$-cell $G/L\times D^{p-1}$, and compose~\eqref{eq:16} with the quotient
map
  \begin{equation}\label{eq:17}
     E^{hG}_{k-1}(Z)\longrightarrow E^{hG}_{k-1}(Z,Z')\cong
     E^{hG}_{k-1}(G/L\times S^{p-1})\cong E^{p-k}(BL). 
  \end{equation}
If the composite is nonzero, which means the boundary of the $p$-cell
attached in~\eqref{eq:15} intersects the $(p-1)$-cell in~\eqref{eq:17}, then
since the stabilizer subgroup can only increase by taking the boundary, we
must have $L'\subset L$.  The composite $E^{p-k}(BL')\to E^{p-k}(BL)$ is the
\emph{transfer}, the pushforward along the finite cover $BL'\to BL$ with
fiber~$L'/L$.   
 
The Atiyah-Hirzebruch spectral sequence is obtained by filtering a $G$-CW
complex by its skeleta and systematizing the argument above.  Suppose $Y$ is
the complement of a subcomplex $Y_0\subset \bY$ of a finite $G$-CW complex.
Then the $E^1$-page of the spectral sequence is the \emph{Bredon homology} of
$(\bY,Y_0)$ with coefficients in the covariant functor on the orbit category
of~$G$ with values in $\ZZ$-graded abelian groups whose component in
degree~$q$ at~$G/L$ is $E^{-q }(BL)$, which is the language used to describe
the systematization of the previous paragraph.  In degree~$-q$ the
coefficient group is the group of invertible topological phases of symmetry
type~$(H\times L,\rho \times e)$ in (spatial) dimension~$q$;
see~\eqref{eq:9}.  This is the $E^1$-page contribution of an equivariant
$p$-cell $e^p\times G/L$.  The spectral sequence converges to an associated
graded of $E^{hG}_0(\bY,Y_0)$.

The differential $d^1$~is the composition of the usual equivariant cellular
boundary map with a transfer map, the latter nontrivial in case the
stabilizer group~$L$ of a $(p-1)$-cell~$e$ is strictly larger than the
stabilizer group~$L'$ of a $p$-cell~$e'$ whose boundary rel the
$(p-2)$-skeleton maps with nontrivial degree to~$e$.  Assume $G$~is finite.
The transfer\footnote{We assume an inclusion $L'\subset L$; an inclusion into
a conjugate ~$gLg\inv $ is then composition with an automorphism.}
$E^{-d}(BL')\longrightarrow E^{-d}(BL)$ has a field-theoretic interpretation
as a map from $(d+1)$-dimensional theories of $H$-manifolds equipped with a
principal $L'$-bundle to $(d+1)$-dimensional theories of $H$-manifolds
equipped with a principal $L$-bundle.  If $M$~is a manifold (bordism)
equipped with a principal $L$-bundle $P\to M$, then a section of the
associated fiber bundle $P/L'\to M$ with fiber~$L/L'$ is equivalent to a
reduction of~$P\to M$ to structure group~$L'\subset L$.  The evaluation of
the transfer of~$F$ on~$(M,P)$ is the (tensor) product over sections of
$P/L'\to M$ of the values of the theory~$F$.  In general sections only exist
locally, so we must use the extended locality of these field theories to
compute the transfer.

We remark that there is a similar spectral sequence if $G$~is a compact Lie
group.  See~\cite{SXG} for further information about the Atiyah-Hirzebruch
spectral sequence in this context.

  \subsection{Fermionic phases on~$\EE^3$ with a half-turn}\label{subsec:1.8}

By way of illustration we now turn to the classification of phases on
$\E^{3}$ which are symmetric with respect to the involution $(x,y,z)\mapsto
(x,-y,-z)$.  The one point compactification of $\E^{3}$ is the equivariant
sphere $S^{1+2\sigma}$, where $\sigma$ is the real sign representation.  The
symmetry type~$(H,\rho)$ has $H$~ the infinite $\spin$ group, and in this
case we may identify $\mtspin$ with $\mspin$.  Applying Ansatz~\ref{thm:3}
in the form~\eqref{eq:10}, we determine the group of equivariant phases to be 
  \begin{equation}\label{eq:19}
     [\mspin,\Sigma^{2}I_{\Z}\wedge S^{1+2\sigma}]^{h\Z/2}. 
  \end{equation}
We compute this group is three ways.

   \subsubsection{First method}\label{subsubsec:1.8.1}
 
 Apply~\eqref{eq:10} with $d=3$ and $\lambda =1+2\sigma $ to compute
\eqref{eq:19} as
  \begin{equation}\label{eq:11}
     [\Sigma ^{2-2\sigma }\RP^{\infty}\wedge \MSpin\,,\,\Sigma ^5\IZ]\cong
     [\Sigma ^{2-2\sigma }\RP^{\infty}\,,\,\Sigma \,ko\langle 0\cdots 4 \rangle]. 
  \end{equation}
Here we use the Anderson-Brown-Peterson~\cite{ABP} decomposition of~$\MSpin$,
in which the leading term is~$ko$ and higher terms do not appear since
$\Sigma ^5\IZ$ has vanishing homotopy groups above dimension~5; we also use
the Anderson self-duality of~$ko$ (with a shift of~4)~\cite{HS}.  Note
$\Sigma ^{2-2\sigma }\RP^{\infty}_+$ is the Thom spectrum
$\Thom(\RP^{\infty};\underline{\RR^2}-L^{\oplus 2})$, where
$L\to\RP^{\infty}$ is the tautological real line bundle.  Let $U$~be the Thom
class of $\underline{\RR^2}-L^{\oplus 2}\to\RP^{\infty}$, $\bU$~its mod~2
reduction, and $a\in H^1(\RP^{\infty};\zt)$ the generator.  The right hand
side of~\eqref{eq:11} can be computed from the nonequivariant
Atiyah-Hirzebruch \emph{cohomology} spectral sequence 
  \begin{equation}\label{eq:12}
     E_2^{p,q}\cong H^p\bigl(\Sigma ^{2-2\sigma }\RP^{\infty};\,ko\langle 0\cdots 4
     \rangle^{q}(\pt)\bigr)\Longrightarrow   
     [\Sigma ^{2-2\sigma }\RP^{\infty}\,,\,\Sigma ^{p+q}\,ko\langle 0\cdots 4 \rangle].
  \end{equation}
The contributions in total degree~$1$ come from $E^{2,-1}_2\cong \zt\cdot \bU
a^2$ and $E^{3,-2}_2\cong \zt\cdot \bU a^3$, which are killed respectively by
$d_2(\bU)=Sq^2(\bU)$ from~$E^{0,0}_2$ and $d_2(\bU a) =Sq^2(\bU a)$
from~$E^{1,-1}_2$.  (Observe $Sq^k(\bU)=\bU w_k(\underline{\RR^2}-L^{\oplus
2})$.)  Thus the group~\eqref{eq:19} of phases vanishes in this case.

   \subsubsection{Second method}\label{subsubsec:1.8.2}
 Decompose $S^{1+2\sigma}$ into pieces of
fixed isotropy and make use of the methodology described
in~\S\ref{subsubsec:1.7.1}. 
The first step is to write 
\[
S^{1+2\sigma} = S^{1}\wedge S^{2\sigma}
\]
and
\[
[M\spin,\Sigma^{2}\IZ\wedge S^{1+2\sigma}]^{h\Z/2}=
[M\spin,\Sigma^{3}\IZ\wedge S^{2\sigma}]^{h\Z/2}.
\]
Now $S^{2\sigma}$ is the unreduced suspension of the unit sphere
$S(2\sigma)\subset \R^{2\sigma}$ so there is a cofibration sequence
of pointed $\Z/2$-spaces (or spectra)
\[
S(2\sigma)_{+}\to S^{0} \to S^{2\sigma}
\]
and an exact sequence
\begin{multline}
\label{eq:m3}
[\mspin,\Sigma^{3}\IZ\wedge S(2\sigma)_{+}]^{h\Z/2}\to 
[\mspin,\Sigma^{3}\IZ]^{h\Z/2}  \to 
[\mspin,\Sigma^{3}\IZ\wedge S^{2\sigma}]^{h\Z/2} \\ \to 
[\mspin,\Sigma^{4}\IZ\wedge S(2\sigma)_{+}]^{h\Z/2}  \to 
[\mspin,\Sigma^{4}\IZ]^{h\Z/2}.
\end{multline}
We will check that
\begin{equation}
\label{eq:m8}
[\mspin,\Sigma^{3}\IZ\wedge S(2\sigma)_{+}]^{h\Z/2}\to 
[\mspin,\Sigma^{3}\IZ]^{h\Z/2} 
\end{equation}
is an epimorphism and
\begin{equation}
\label{eq:m9}
[\mspin,\Sigma^{4}\IZ\wedge S(2\sigma)_{+}]^{h\Z/2}\to 
[\mspin,\Sigma^{4}\IZ]^{h\Z/2} 
\end{equation}
is a monomorphism, from which we deduce
\[
[\mspin,\Sigma^{3}\IZ\wedge S^{2\sigma}]^{h\Z/2} =0.
\]
 This implies that there is only one phase on $\E^{3}$---the trivial
phase---which is symmetric with respect to the involution $(x,y,z)\mapsto
(x,-y,-z)$.

To evaluate~\eqref{eq:m8} and~\eqref{eq:m9} note the orbit space
$S(2\sigma)/(Z/2)$ is just $\rp^{1}=S^{1}$ so from the Adams
isomorphism~\eqref{eq:20} we have  
\[
[\mspin, \Sigma^{k}\IZ\wedge S(2\sigma)_{+}]^{h\Z/2}
\approx [\mspin, \Sigma^{k}\IZ\wedge \rp^{1}_{+}] .
\]
The composition
\[
[\mspin,\Sigma^{3}\IZ\wedge S(2\sigma)_{+}]^{h\Z/2}\to 
[\mspin,\Sigma^{3}\IZ]^{h\Z/2} \to [\mspin,\Sigma^{3}\IZ]
\]
is the map induced by the transfer map of spectra 
  \begin{equation}\label{eq:23}
     \rp^{1}_{+}\to S^{0}. 
  \end{equation}
A choice of base point in $\rp^{1}$ gives a weak equivalence 
\begin{equation}
\label{eq:m10}
\rp^{1}_{+}\xrightarrow{\approx}{} S^{1}\vee S^{0}.
\end{equation}
Since $\rp^{1}$ is path connected, the homotopy class of this map is
independent of this choice.

The following can be proved using standard methods.

\begin{prop}
\label{thm:m8}
With respect to the decomposition~\eqref{eq:m10} the transfer map 
\[
\rp^{1}_{+}\to S^{0}
\]
has components 
\begin{align*}
\eta:S^{1}&\to S^{0} \\
2:S^{0}&\to S^{0} .
\end{align*}
in which $\eta\in\pi_{1}S^{0}=\Z/2$ is the non-trivial element.
\qed
\end{prop}

Using the fact that the Atiyah-Bott-Shapiro map $\mspin\to ko$ is an
equivalence up to dimension $8$, and the isomorphisms
\begin{align*}
[\mspin,\Sigma^{k}I\Z\wedge S(2\sigma)_{+}]^{h\Z/2} &\approx 
[\mspin,\Sigma^{k}I\Z\wedge \rp^{1}_{+}] \\ 
&\approx [\mspin,\Sigma^{k+1}I\Z] \oplus [\mspin,\Sigma^{k}I\Z]  \\
[\mspin,\Sigma^{k}I\Z]^{h\Z/2} &\approx
[\mspin\wedge B\Z/2_{+},\Sigma^{k}I\Z] \\ 
&\approx
[\mspin\wedge B\Z/2,\Sigma^{k}I\Z]\oplus
[\mspin,\Sigma^{k}I\Z]
\end{align*}
one extracts the following table of values
  \begin{equation}\label{eq:25}
  \begin{gathered}
    \includegraphics[scale=1.2]{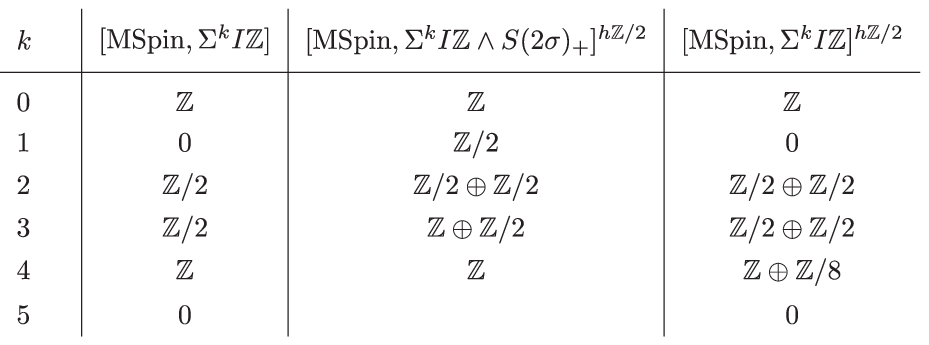} 
  \end{gathered}
  \end{equation}
as well as the fact that multiplication by the non-zero element
$\eta\in \pi_{1}S^{0}$ is the non-trivial map 
\[
[\mspin,\Sigma^{k}I\Z]\to 
[\mspin,\Sigma^{k-1}I\Z]
\]
when $k=4$ or $3$.     

By Remark~\ref{rem:5}, homomorphisms
\begin{align*}
[\mspin,\Sigma^{k}I\Z\wedge (\Z/2)_{+}]^{h\Z/2} &\to
[\mspin,\Sigma^{k}I\Z]^{h\Z/2} \\
[\mspin,\Sigma^{k}I\Z]^{h\Z/2} &\to
[\mspin\wedge (\Z/2)_{+},\Sigma^{k}I\Z]^{h\Z/2} \\
\end{align*}
induced by the map $\Z/2\to\text{pt}$ can be identified with the maps
\begin{align*}
[\mspin,\Sigma^{k}I\Z] &\to [\mspin\wedge B\Z/2_{+},\Sigma^{k}I\Z] \\
[\mspin\wedge B\Z/2_{+},\Sigma^{k}I\Z] &\to [\mspin,\Sigma^{k}I\Z] \\
\end{align*}
induced by the transfer map $B\Z/2_{+}\to S^{0}$, and the inclusion
map $S^{0}\to B\Z/2_{+}$ associated to a choice of point in $B\Z/2$.
The effect of the transfer map is given by the following table
\begin{center}
\includegraphics[scale=1.2]{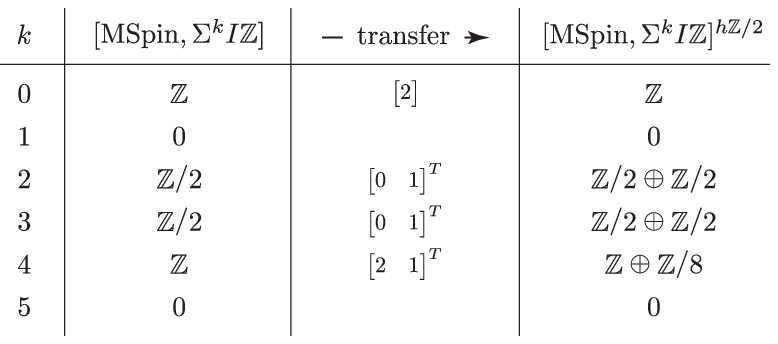}
\end{center}

With these values, and Proposition~\ref{thm:m8} the map
\[
[\mspin,\Sigma^{3}\IZ\wedge S(2\sigma)_{+}]^{h\Z/2}\to 
[\mspin,\Sigma^{3}\IZ]^{h\Z/2} 
\]
becomes 
  \begin{equation}\label{eq:24}
     \Z\oplus\Z/2\xrightarrow{\begin{bmatrix}\ast & 1\\ 1&0\end{bmatrix}}
     \Z/2\oplus \Z/2 
  \end{equation}
which is indeed an epimorphism, while
\[
[\mspin,\Sigma^{4}\IZ\wedge S(2\sigma)_{+}]^{h\Z/2}\to 
[\mspin,\Sigma^{4}\IZ]^{h\Z/2} 
\]
becomes
\[
\Z\xrightarrow{\begin{bmatrix} 
2 \\ \ast\end{bmatrix}} \Z\oplus \Z/8
\]
which is a monomorphism.

   \subsubsection{Third method}\label{subsubsec:1.8.3}
 A lecture by Mike Hermele based on~\cite{HSHH} suggested to us that the
equivariant Atiyah-Hirzebruch homology spectral sequence has a physical
interpretation in this context; here we describe how this spectral sequence
plays out to kill the relevant group.  See also~\cite{SXG} for many worked
examples using this spectral sequence.  We refer to~\S\ref{subsubsec:1.7.2}
for an exposition of the equivariant Atiyah-Hirzebruch spectral sequence.  In
the case of equivariant phases on $\E^{3}$, we use the equivariant cell
decomposition \[ S^{1+2\sigma} = S^{1}\;\cup\; \Z/2\times e^{2} \;\cup\;
\Z/2\times e^{3} \] of the one-point compactification of~$\EE^3$, the
appropriate representation sphere.  Using the table~\eqref{eq:25}, the
spectral sequence works out to be \begin{center}
\includegraphics[]{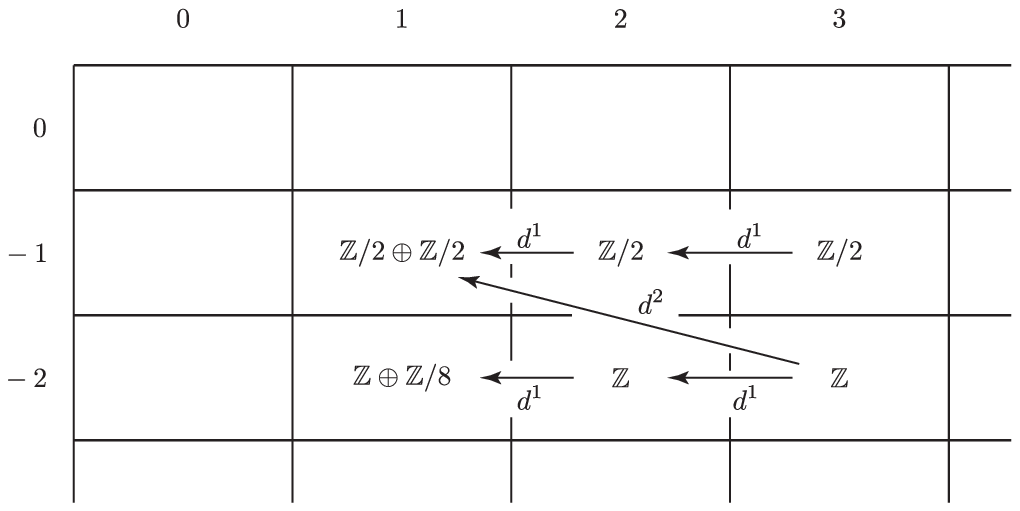}
\end{center}
The $(1,-d)$ entry is the group of invertible $(d+1)$-dimensional fermionic
phases with internal symmetry group~$\zt$ (the stabilizer group of the
1-cell), and the $(p,-d)$ entry for~$p=2,3$ is the group of invertible
$(d+1)$-dimensional fermionic phases.  The group~\eqref{eq:18} of interest is
the homology in degree~0.

  \begin{clm}[]\label{thm:11}
 The spectral sequence scorecard in degree~0 is:

 \begin{enumerate}[label=\textnormal{(\roman*)}]

 \item the differential $d^1\:E^1_{2,-1}\to E^1_{1,-1}$ hits a
$\zt$-subgroup;

 \item $d^1\:E^1_{2,-2}\to E^1_{1,-2}$ is injective; and

 \item the differential $d^2\:E^2_{3,-2}\to E^2_{1,-1}$ is onto the
remaining~$\zt$.

 \end{enumerate}
  \end{clm} 

  \begin{proof}
 The group~$E^1_{1,-1}$ of invertible topological
phases of spin $2$-manifolds~$X$ equipped with a double cover $Q\to X$ may be
described in terms of partition functions.  Recall that a spin structure on a
closed 2-manifold~$X$ gives a quadratic refinement~$q_X$ of the intersection
pairing on~$H^1(X;\zt)$, and $q_X$~ has an Arf invariant~$\Arf(q_X)\in \zt$.
The equivalence class of a double cover $Q\to X$ lives in $H^1(X;\zt)$.  The
four possible partition functions are $1$, $(-1)^{\Arf(q_X)}$,
$(-1)^{\Arf(q_X+Q)}$, and $(-1)^{q_X(Q)}$.  A more precise version of~(i) is:
the first differential $d^1\:E^1_{2,-1}\to E^1_{1,-1}$ maps the second of
these, which is a theory on spin manifolds without a double cover, onto the
last of these.\footnote{which Mike Hermele calls a ``fermionic AKLT state''}
We can compute that from the transfer as follows.  Let the target 2-groupoid
for these extended field theories be the Morita category of central simple
complex superalgebras equipped with a $\zt$-action.  The four theories
evaluate on a point respectively to $\CC$ with trivial involution, the
Clifford algebra~$A=\Cliff_1^{\CC}$ with trivial involution, the algebra~$A$
with nontrivial involution, and $\CC$ with nontrivial involution.  The
transfer maps the second of these to~$A\otimes A$ with the involution
exchanging the factors, and this is Morita equivalent to~$\CC$ with
nontrivial involution.  This proves~(i).  Claim~(ii) is straightforward: the
differential $d^1\:E^1_{2,-2}\to E^1_{1,-2}$ does not involve a transfer, so
reduces to the cellular differential.  The differential in~(iii) is induced
by the transfer~\eqref{eq:23}, and was worked out in~\eqref{eq:24}.
  \end{proof}

 \bibliographystyle{hyperamsalpha} 
\providecommand{\bysame}{\leavevmode\hbox to3em{\hrulefill}\thinspace}
\providecommand{\MR}{\relax\ifhmode\unskip\space\fi MR }
\providecommand{\MRhref}[2]{%
  \href{http://www.ams.org/mathscinet-getitem?mr=#1}{#2}
}
\providecommand{\href}[2]{#2}

  \end{document}